\documentclass[letterpaper, 10 pt, conference]{ieeeconf}  

\IEEEoverridecommandlockouts                              

\usepackage{graphicx}
\usepackage{dsfont}
\usepackage{algorithm} 
\usepackage{algpseudocode} 
\usepackage{amsmath,amssymb,amsfonts}
\usepackage{cleveref}
\DeclareMathOperator*{\argmax}{arg\,max}

\newtheorem{theorem}{Theorem}[section]
\newtheorem{lemma}[theorem]{Lemma}

\usepackage[markup=bfit, deletedmarkup=sout, authormarkup=superscript]{changes}
\definechangesauthor[name={jk}, color=orange]{JK}

\definechangesauthor[name={ar}, color=magenta]{AR}

\definechangesauthor[name={lc}, color=red]{LC}

\definechangesauthor[name={chil}, color=red]{ChiL}

\title{\LARGE \bf
ROMA-iQSS: An Objective Alignment Approach via State-Based Value Learning and ROund-Robin Multi-Agent Scheduling 
}

\author{Chi-Hui Lin, Joewie J. Koh, Alessandro Roncone, Lijun Chen
\thanks{The authors are with the Department of Computer Science, University of Colorado, Boulder, Colorado, USA. Emails: {\tt\small \{firstname.lastname\}@colorado.edu}}%
\thanks{This work was supported by Army Research Laboratory under grant \#W911NF-21-2-0126.
Alessandro Roncone is with Lab0, Inc.}
}

\begin{document}

    \maketitle
    \thispagestyle{empty}
    \pagestyle{empty}

    \begin{abstract}  
        Effective multi-agent collaboration is imperative for solving complex, distributed problems. In this context, two key challenges must be addressed: first, autonomously identifying optimal objectives for collective outcomes; second, aligning these objectives among agents. Traditional frameworks, often reliant on centralized learning, struggle with scalability and efficiency in large multi-agent systems. To overcome these issues, we introduce a decentralized state-based value learning algorithm that enables agents to independently discover optimal states. Furthermore, we introduce a novel mechanism for multi-agent interaction, wherein less proficient agents follow and adopt policies from more experienced ones, thereby indirectly guiding their learning process. Our theoretical analysis shows that our approach leads decentralized agents to an optimal collective policy. Empirical experiments further demonstrate that our method outperforms existing decentralized state-based and action-based value learning strategies by effectively identifying and aligning optimal objectives.
    \end{abstract}

    \section{INTRODUCTION}
        Reinforcement Learning (RL) enables autonomous agents to make informed decisions by leveraging past experiences to anticipate future rewards in a variety of contexts.
        However, despite its notable successes in optimizing the behavior of single- or few-agent systems, RL's limitations (e.g.; sample inefficiency, high variance, out-of-distribution behaviors) become more evident when transitioning to multi-agent settings.
        These challenges highlight the critical need for effective inter-agent coordination, a requirement that is particularly important for applications such as
        robotic warehouse management \cite{papoudakis2021benchmarking,mrc}, urban traffic light systems \cite{tsc, Wei_2019}, and human-robot collaboration \cite{brawer2023interactive,tung2024workspaceopt}.

        For optimal coordination, agents face two overarching challenges.
        First, they must identify which objectives will maximize collective utility.
        Second, they must achieve goal alignment with their teammates to avoid counterproductive outcomes.
        Traditionally, the above challenges were addressed by methods that employ a \textsl{centralized} learning approach. 
        Specifically, these methods aggregate experiences from all agents to perform joint strategy optimization, a technique that proved effective in optimizing group outcomes
        \cite{son2019qtran, gupta2017cooperative, lowe2020multiagent, iqbal2019actorattentioncritic}. 
        However, these approaches face three critical challenges: 
        i) they incur exponential computational costs as the number of agents increases,
        ii) they necessitate the global sharing of individual agent policies, which might not be feasible in practice and limits the applicability of such methods, 
        and iii) it is impractical to maintain continuous, high-bandwidth communication with a central controller in dynamic environments, such as self-driving vehicles.
        Collectively, these limitations inevitably limit the scalability and robustness of centralized training paradigms.
        
        \begin{figure}
            \centering
            \includegraphics[width=0.88\columnwidth]{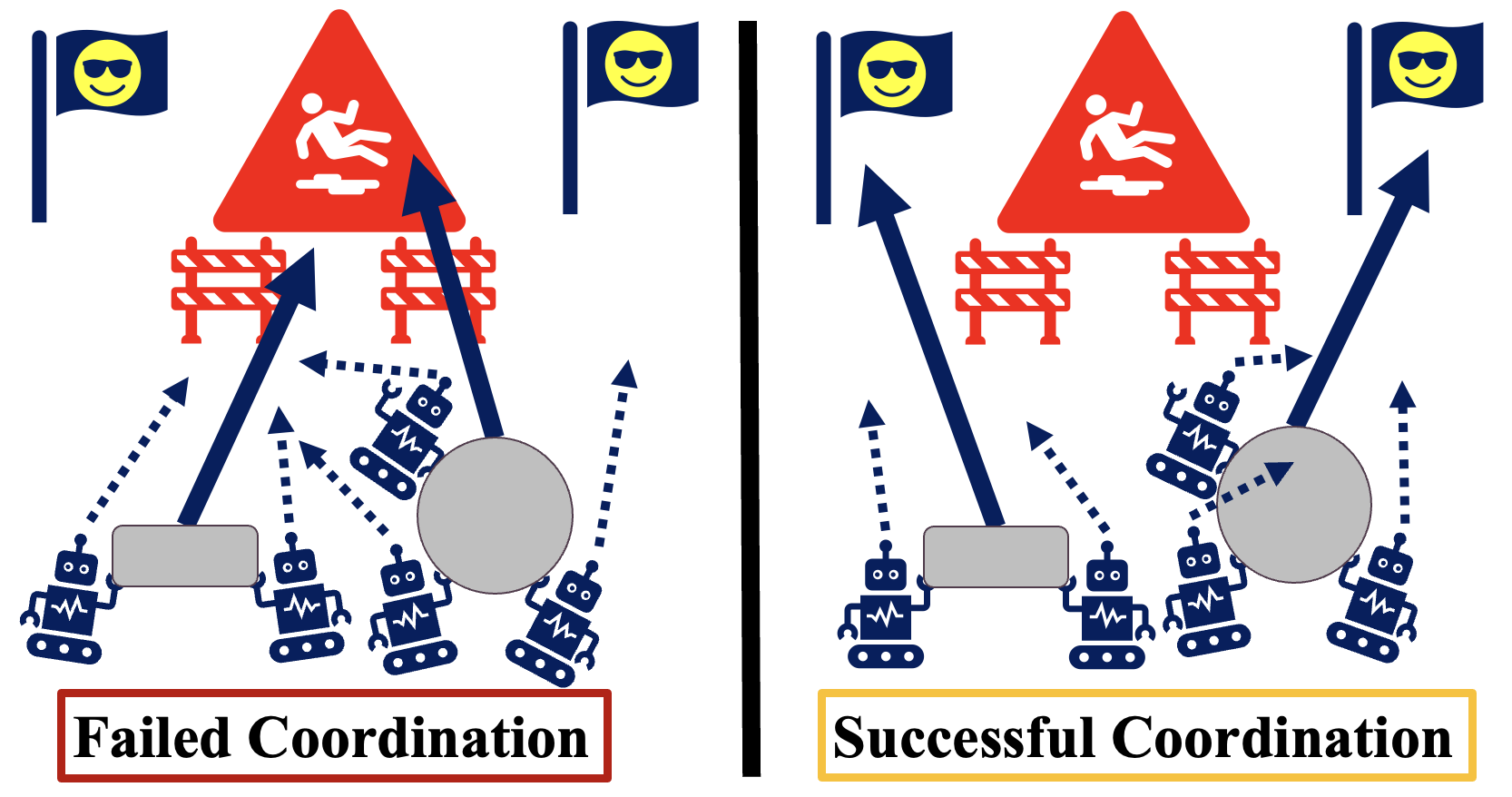}
            
            \caption{Motivating example: five agents are tasked with a collaborative transport problem where a large object needs to be moved to designated locations marked by flags. In the case of failed coordination, the left subgroup (two agents), struggles to identify the optimal objectives, while the right subgroup (three agents), identifies the objectives but fails to align their efforts effectively, leading to an undesired final location for the object. In contrast, successful coordination manifests when all agents not only identify the optimal objectives but also achieve precise alignment in their efforts, culminating in the object reaching its intended location.}
            \label{figurelabel}
        \end{figure}

        Conversely, \textsl{decentralized} learning approaches offer a compelling alternative to centralized training, allowing each agent to estimate expected returns autonomously based on individual experiences. Crucially, these algorithms preclude direct inter-agent communication and restrict access to the policies of other agents \cite{ming1993multi,dewitt2020independent,palmer2018lenient,4399095, I2Q}. As such, decentralized learning paradigms, such as independent Q-learning \cite{tan1997multi}, are particularly well-suited for highly scalable scenarios. Nevertheless, the limited flow of information creates obstacles in agents' understanding of their peers and their ability to interpret the environment effectively. This limitation frequently causes a lack of alignment in agents' objectives or the selection of suboptimal goals, ultimately resulting in poor coordination.

        Human coordination can thrive even without verbal communication, as silent interactions often suffice for mutual understanding. Motivated by this concept, we tackle the hurdles of limited communication within decentralized systems by examining their interaction patterns. Typically, agents operate concurrently, deriving lessons from these experiences---a method we have termed \textsl{Synchronous MultiAgent Interaction (SMA)}. While this strategy promotes efficient exploration, it also results in a continuously changing environment, complicating the agents' capacity to align their strategies. To counteract this issue, we present \textsl{ROMA}, a \textsl{ROund-Robin MultiAgent Scheduling} protocol designed to streamline information acquisition during interaction for goal alignment. Within this framework, interaction is structured in rounds, with each round permitting only one agent to collect experience from the environment. Additionally, more knowledgeable agents utilize their better policies to guide the agent under consideration toward beneficial outcomes, while less experienced agents are afforded the opportunity for exploratory behaviors. This synergistic approach allows agents to capitalize on the expertise of their more experienced peers, culminating in aligned objectives.

        Beyond aligning objectives, equipping agents with the skills to pinpoint the best goals is also crucial. To overcome the constraints of decentralized algorithms, we propose \textsl{independent QSS learning} (\textsl{iQSS}), an innovative approach to decentralized learning based on state values. iQSS works by assessing the value of the current global state based on potential future states, utilizing insights gained from environmental interactions. As a result, agents can more accurately anticipate the returns of possible future scenarios, guiding them more efficiently towards the optimal objectives. 

        In conclusion, this paper introduces a tightly integrated framework that combines independent state-based value learning, iQSS, with a specialized multi-agent interaction protocol, ROMA. These two components are not stand-alone solutions; rather, they operate synergistically to enable agents to pinpoint optimal states and synchronizing their objectives. Theoretical analysis reveals iQSS helps agents converge on effective policies for optimal states, while ROMA coordinates their efforts to a common goal.
        Our empirical studies, featuring multi-stage coordination tasks, demonstrate ROMA-iQSS's superiority over the current state-of-the-art, I2Q \cite{I2Q}, and traditional independent Q learning methods. Crucially, ROMA-iQSS stands out for its dual capability to identify optimal objectives and ensure goal alignment among agents.

    \section{Preliminaries}
        \subsection{Coordination Problem}
            We aim to solve the following collaborative problem, represented by a 7--tuple Markov Decision Process,  $\{ K, S,(A_k)_{k\in K}, T,r,(\pi_k)_{k\in K}, \gamma \}$:
    
            \begin{itemize}
                \item $\mathcal{K} = \{ 1, \dots, k \}$ is the set of indexes of all agents.
                \item $S$ is the environment state space.
                \item $A_k$ is the discrete action space for agent $k$, while $A_{-k}$ denotes the action space for all agents except agent $k$.
                \item $s_{(t+1)} = T(s_{(t)}, a_{(t)})$ is the environment transition function. To clarify, time step indices are indicated with parentheses, such as $a_{(t)}$ for joint action at time $t$, and agent-ID indices are shown without parentheses, like $a_{k}$ for agent $k$’s action.
                \item $r$ is the shared reward received by all agents $k \in K$.
                \item $\pi_k$, belongs to a deterministic policy class $\Pi_k$, generates actions $a_k\in A_k$ for agent $k$ under environment states.
                \item $\gamma$ is a discount factor.
            \end{itemize}
            Our main goal is to develop a learning method to enable each agent to independently learn its optimal policy, $\pi_k$, using its own experiences, without the explicit knowledge of policies shared by other agents. More specifically, the joint of all individual optimal policies should be an optimal joint policy $\pi^*$ that maximizes the cumulative discounted shared rewards $r$ with a discount factor denoted as $\gamma$. These rewards depends on the current state, $s_{(t)}$, and the next state, $s_{(t+1)}$, arising from a deterministic transition function, $T$. This function uses the current state and the action derived from the policy, $\pi(s_{(t)})$, to generate the next state.
            \begin{align*}\label{obj}
                \pi^*  &= \argmax_{\pi \in \Pi} \mathbb{E}[\sum_{s_{(0)} \in S}\sum_{t\geq 0}\gamma^t r(s_{(t)},s_{(t+1)})], \\
                &s_{(t+1)} = T(s_{(t)},\pi(s_{(t)})), \\       
                &\Pi       =\Pi_1\times\Pi_2\times...\times\Pi_{|K|}
            \end{align*} 
            Additionally, agents collect their own experiences when interacting with other agents in the environment. During an agent's interaction, it receives multiple pieces of information, each represented by a tuple $(s, a_k, s', r)$. Each tuple contains the current environment state $s$, an agent's own action $a_k$, a subsequent environment state $s'$, and a shared reward $r$.
            
        \subsection{Overview of Value Learning Methods}
            Centralized Q learning provides a method to generate an optimal policy that maximizes a group's profits, but it suffers from scalability. On the other hand, independent Q learning enables agents to learn independently. However, it is challenging for it to learn an optimal policy due to its partial understanding of the environment.
    
            \subsubsection{Centralized Q Learning (cenQ)}
                In the context of our coordination problem,  the utilization of centralized Q-learning necessitates the presence of a central controller.  The central controller accesses the experience of all agents and leverages it to learn the joint strategy for all agents.
 
                Specifically, cenQ maintains a value function, representing an expected long-term return associated with the environment state and joint action, $a$, denoted as $Q(s, a)$, where $a=(a_1, a_2, ..., a _{|K|})$ and $a \in A$, where $ A=A_1 \times A_2 \times ... \times A_{|K|}$. This value function is learned through a temporal difference learning process, incorporating $\alpha$ as the learning rate, $\gamma$ as the discount factor, and $s'$ as the subsequent environment state, determined by the environment transition function $T(s, a)$.
                \begin{equation}\label{qSA}
                    \begin{aligned}
                        Q(s, a) \leftarrow
                                &(1-\alpha)Q(s,a) \\
                                &+\alpha(r + \gamma \max_{a'\in A} Q(s', a'))
                    \end{aligned}
                \end{equation}
                              
                Temporal difference learning in \cref{qSA} is well-known for yielding the optimal Q-value, $Q^*(s, a)$, upon convergence \cite{watkins1992q}. 
                \begin{align*}
                    Q^*(s, a) = r(s,T(s,a))+\gamma \max_{a'\in A}Q^*(s', a')
                \end{align*}
            
                With the value function $Q^*(s, a)$, The centralized controller can further induce the optimal policy, $\pi^*$:
                \begin{align*}
                    \pi^*(s) = \argmax_{a \in A}Q^*(s, a) 
                \end{align*}

                However, even with the promise of optimal convergence, the preference for cenQ diminishes when applied to extensive multi-agent systems due to inherent scalability concerns.
            \subsubsection{Independent Q Learning (indQ)}\label{sec:IIndependent-Q-Learning}
                Independent Q Learning becomes a preferable option in multi-agent systems due to its superior scalability. This permits an agent to estimate its value based on individual experiences. In technical terms, each agent, $k \in K$, maintains the value $Q_k$, representing the expected return associated with its own action and environment state, through temporal difference learning:
                \begin{equation}\label{TD_k}
                    \begin{aligned}
                        Q_k(s, a_k) \leftarrow
                                &(1-\alpha)Q_k(s,a_k) \\
                                &+\alpha(r + \gamma \max_{a_k'\in A_k} Q_k(s', a_k'))
                    \end{aligned}
                \end{equation}

                In \cref{TD_k}, the subsequent state $s'$ is determined by the transition function, $T(s, a)$. However, an independent agent $k$ is limited to accessing the information of its own action denoted as $a_k$, instead of the complete action set $a$. It can only access its action information, $a_k$. Therefore, the estimated value, $Q_k(s,a_k)$, depends on its observed transition, $P(s'|s,a_k)$, which is a distribution conditioned on its action and environment state.
                \begin{align*}
                    &Q_k^*(s, a_k) = \mathbb{E}_{s' \sim P(s'|s,a_k)}[r(s,s')+\gamma \max_{a_k'\in A_k}Q_k^*(s', a_k')]
                \end{align*}
                The observed transition, $P(s'|s,a_k)$, in an environment with deterministic transition, $T(s,a_k,a_{-k})$ can be represented as: 
                \begin{align*}
                    P(s'|s,a_k) 
                        & = \sum_{a_{-k} \in A_{-k}} P(s'|s,a_k,a_{-k})P(a_{-k}|s,a_k) \\
                        & = \sum_{a_{-k} \in A_{-k}} \mathds{1}_{s'=T(s,a_k,a_{-k})}P(a_{-k}|s,a_k) 
                \end{align*}
                \begin{align*}
                    \mathds{1}_{s'=T(s,a_k,a_{-k})} = 
                        \begin{cases}
                            1 &, if\ s'=T(s,a_k,a_{-k}) \\
                            0 &, otherwise
                        \end{cases}
                \end{align*}
                 Therefore, $P(s'|s,a_k)$ is dependent on the responses, $P(a_{-k}|s,a_k)$, from other agents. If they keep changing their response strategy, the individual observed transition can be non-stationary, making the convergence on the value function, $Q_k(s,a_k)$, challenging. 

                Furthermore, if other agents maintain stable but suboptimal strategies, the agent could learn a suboptimal value. This, in turn, results in its derived policy, $\pi^*_k$, also being suboptimal.
                \begin{align*}
                    \pi^*_k(s) = \argmax_{a_k \in A_k}Q^*_k(s, a_k) 
                \end{align*}
                
        \subsubsection{Centralized QSS Learning (cQSS)}
            QSS learning distinguishes itself as a state-based value learning. Unlike action-based methods, its return estimate is not directly dependent on the actions. Instead, it is determined by only a state and a subsequent state.
            
            When designing a centralized QSS learning, like centralized Q learning, it also necessitates including a central controller that can access the information of all agents. However, instead of estimating the state-action value, it approximates the state-state value, $Q(s, s')$, an expected return associated with a state $s \in S$, and a neighboring state $s' \in N(s)$, where $N(s)$ is a set of possible subsequent states, observed by agents during the interaction.
            \begin{align*}
                N(s) &=  \{s'|P(s'|s) > 0\}\\
                    &= \{s'|\mathds{1}_{s'=T(s,a)}>0, \forall a\in A\} 
            \end{align*}

            Through the temporal difference learning, the central controller learns the converged state-state value $Q^*(s, s')$ \cite{QSS}. 
            \begin{align*}
                Q(s, s') \leftarrow
                        &(1-\alpha)Q(s,s) \\
                        &+\alpha(r + \gamma \max_{s''\in N(s')} Q(s', s''))
            \end{align*}
            \begin{align*}
                Q^*(s, s') = r(s,s')+\gamma \max_{s''\in N(s')}Q^*(s', s'')         
            \end{align*}
            Edwards et al. show the equivalence between $Q^*(s, a)$ and $Q^*(s, s')$, where the subsequent state $s'$ is determined by the transition $T(s, a)$ \cite{QSS}:
            \begin{align*}
                Q^{ss*}(s, a) = Q^*(s,T(s,a))
            \end{align*}
            We can utilize the equivalence to let the central controller induce its optimal state-action value $Q^{ss*}(s, a)$ and further induce the policy $\pi_{ss}^*$, which is guaranteed to be optimal.
            \begin{align*}
                \pi^{ss*}(s) = \argmax_{a \in A}Q^{ss*}(s, a) 
            \end{align*}
    
    \section{Methods}
    \subsection{Independent QSS Value Learning (iQSS)}
        Although centralized QSS learning can learn the optimal policy, similar to other centralized approaches, it is also not a scalable choice. Therefore, we develop an independent state-based value learning method, independent QSS learning.

        While centralized QSS learning enables a central controller to learn the state-state value, independent QSS learning enables each agent to learn its own value. Each agent learns the value through the temporal difference learning, used by the centralized QSS controller. Due to the convergence property \cite{watkins1992q, QSS}, each agent can ultimately learn the unique optimal state-state value $Q^*(s,s')$. Nevertheless, since agents learn the value independently, we represent their values, respectively, with an index $k$.
        \begin{equation}\label{TD-ind-QSS}
            \begin{aligned}
                Q_k(s, s) \leftarrow
                        &(1-\alpha)Q_k(s,s) \\
                        &+\alpha(r + \gamma \max_{s''\in N(s')} Q_k(s', s'')).
            \end{aligned}
        \end{equation}

        \begin{equation}\label{optQssk}
            \begin{aligned}
                Q^*(s, s')  &= Q_k^*(s, s') \\
                            &= r(s,s')+\gamma \max_{s''\in N(s')}Q_k^*(s', s'')\\        
            \end{aligned}
        \end{equation}

        Like cQSS inducing its action-based value, iQSS agents can also utilize their states-based value to induce their action-based value function. However, unlike the centralized approach, an independent agent cannot access the information of actions, made by others. Therefore, an iQSS agent maintains a value, $Q_k^{ss*}(s, a_k)$ that takes only its own action and the environment state. It induces its state-action value $Q_k^{ss*}(s, a_k)$ by maximizing the state-state value over a set of subsequent states $N(s,a_k)$, which is observed and recorded during its experience collection stage. 

        \begin{equation}\label{QSAk-from-QSSk}
            \begin{aligned}
                Q_k^{ssa*}(s, a_k) =
                        \max_{s' \in N(s,a_k)}Q_k^*(s,s')
            \end{aligned}
        \end{equation}
        
        \begin{equation} \label{N_s_ak}
            \begin{aligned}
                N(s,a_k) &=  \{s'|P(s'|s,a_k) > 0\}
            \end{aligned}
        \end{equation}

        With the converged value, iQSS agents can then induce their policies:
        \begin{align*}
            \pi^{ss*}_k(s) = \argmax_{a_k \in A_k}Q^{ssa*}_k(s, a_k) 
        \end{align*}

        Although the individual policies $\pi^{ss*}_k(s)$ were induced without considering the action information of other agents, we show that the joint policies still have the equivalence to the optimal policy $\pi^{*}(s)$ under the assumption in Theorem \ref{pi_ss_k is optimal}. We later show the assumption is true in Lemma \ref{the set equivalence}.

        \begin{theorem}\label{pi_ss_k is optimal}
            $\pi^{ss*}_k$ is optimal under the set equivalence assumption:
            \begin{align*}
                & \{ s'|\ \exists a_{-k}\in A_{-k}, \mathds{1}_{s'=T(s,a_k,a_{-k})}>0 \} \\
                & = \{ s'| P(s'|s,a_k)>0 \}
            \end{align*}
        \end{theorem}
        \begin{proof}
            \begin{align*}
                & Q^{ssa*}_k( s, a_k) \\
                by\ \cref{QSAk-from-QSSk} \ \ & = \max_{s'\in N(s,a_k)}Q^{*}_k( s, s') \\
                by\ \cref{optQssk} \ \ & = \max_{s'\in N(s,a_k)}Q^{*}( s, s') \\
                by\ \cref{N_s_ak} \ \ & =  \max_{s'\in \{ s'| P(s'|s,a_k)>0 \}}Q^{*}( s, s') \\
                = & \max_{s'\in \{ s'|\ \exists a_{-k}\in A_{-k}, \mathds{1}_{s'=T(s,a_k,a_{-k})}>0 \}}Q^{*}( s, s') \\
                = & \max_{a^{-k} \in A^{-k} }Q^{*}( s, T(s,a_k,a_{-k}))
            \end{align*}  
            The first three equivalence come from the definitions. The fourth is due to the assumption, mentioned in the Theorem \ref{pi_ss_k is optimal}. For the fifth equivalence, the upper-hand side maximizes state-state value over all subsequent states, which any joint actions containing $a_k$ could, given the current state $s$. The lower-hand side basically does the same thing, but it brings the transition function $T(s,a_k, a_{-k})$ into the state-state value function. Therefore, instead of maximizing over states, it maximizes over all possible joint actions while fixing the action made by agent $k$ to be $a_k$.
            
            Let $a_k^*=\pi^{ss*}_k(s)$.
            \begin{align*}
                & Q^{ssa*}_k( s, a_k^*) \\
                & =  \max_{a^{-k} \in A^{-k} }Q^{*}( s, T(s,a_k^*,a_{-k})) \\
                & =  \max_{a_k\in A_k}\max_{a^{-k} \in A^{-k} }Q^{*}( s, T(s,a_k,a_{-k})) \\
                & =  \max_{(a_k,a_{-k})\in A_k\times A_{-k}}Q^{*}( s, T(s,a_k,a_{-k})) \\
                & \geq \ Q^{*}( s, T(s,a)), \forall a \in A
            \end{align*}
            The first equivalence comes from the final equivalence of the process, we have shown. The second equivalence is from the $\pi^{ss*}_k(s)$ definition. For the third equivalence, the upper-hand side first maximizes the value over joint actions containing $a_k$ for each $a_k\in A_k$ and then maximizes it over $a_k\in A_k$. It considers all the Cartesian product's joint actions, $A_k\times A_{-k}$. The lower-hand side of the fourth equivalence maximizes over the Cartesian product of the action sets. Therefore, the fourth equivalence holds. The final equivalence is because of the equivalence $A= A_k \times A_{-k}$.

            The inequality shows the action generated by the $\pi^{ss*}_k(s)$ is the optimal action when the current state is $s$. The property holds for all states $s \in S$. Therefore, $\pi^{ss*}_k(s)$ is optimal.  
        \end{proof}        
        
        Since Theorem \ref{pi_ss_k is optimal} requires the assumption, we show it is true under a condition in the Lemma \ref{the set equivalence}. In addition, we later create multi-agent interaction schemes to let the condition hold.
        
        \begin{lemma} \label{the set equivalence}
            The set equivalence assumption, mentioned in the Theorem \ref{pi_ss_k is optimal}, holds true under the relation assumption:
            \begin{align*}
                & \exists a_{-k}\in A_{-k}, \mathds{1}_{s'=T(s,a_k,a_{-k})}>0 \Rightarrow P(s'|s,a_k) > 0
            \end{align*}
        \end{lemma}
        \begin{proof}
            We show the equivalence by showing they are subsets of each other. For simplicity of clarification, we represent two sets by $X$ and $Y$. We first prove the set $X$ is a subset of the set $Y$.
            \begin{align*}
                & X = \{ s'| P(s'|s,a_k)>0 \} \\
                & = \{ s'| \sum_{a_{-k} \in A_{-k}} P(s'|s,a_k,a_{-k})P(a_{-k}|s,a_k)>0 \}\\
                & = \{ s'| \sum_{a_{-k} \in A_{-k}} \mathds{1}_{s'=T(s,a_k,a_{-k})}P(a_{-k}|s,a_k)>0 \}\\
                & \subseteq \{ s'| \sum_{a_{-k} \in A_{-k}} \mathds{1}_{s'=T(s,a_k,a_{-k})}>0 \} \\
                & = \{ s'|\ \exists a_{-k}\in A_{-k}, \mathds{1}_{s'=T(s,a_k,a_{-k})}>0 \} = Y
            \end{align*}
            By the law of the total probability, we have the first equivalence. Then, since the transition, $P(s'|s,a_k, a_{-k})$, is deterministic in our environments and represented by $s'=T(s'|s,a_k,a_{-k})$, we have the second equivalence. In addition, we have the subset relation between the third and the fourth line because of the inequality for any $s' \in S$:
            \begin{align*}
                &\sum_{a_{-k} \in A_{-k}} \mathds{1}_{s'=T(s,a_k,a_{-k})}P(a_{-k}|s,a_k) \leq \\
                &\sum_{a_{-k} \in A_{-k}} \mathds{1}_{s'=T(s,a_k,a_{-k})}
            \end{align*}
            The inequality is true because $0\leq P(a_{-k}|s,a_k)\leq 1$. Finally, the final equivalence comes from the relation: 
            \begin{align*}
                &\sum_{a_{-k} \in A_{-k}} \mathds{1}_{s'=T(s,a_k,a_{-k})}>0 \Leftrightarrow \\ &\exists a_{-k}\in A_{-k}, \mathds{1}_{s'=T(s,a_k,a_{-k})}>0
            \end{align*}
            It is true because $\mathds{1}$ is non-negative.

            Moreover, we can prove the set $Y$ is a subset of the set $X$ by the assumption, mentioned in the theorem.
            \begin{align*}
                & Y = \{ s'|\ \exists a_{-k}\in A_{-k}, \mathds{1}_{s'=T(s,a_k,a_{-k})}>0 \} \\
                & \subseteq \{ s'| P(s'|s,a_k)>0 \} = X
            \end{align*}
            Since they are subsets of each, we have the set equivalence in Lemma \ref{the set equivalence}.
        \end{proof}
        \begin{lemma}\label{equality-prob}
            The relation assumption, mentioned in Lemma \ref{the set equivalence}, holds true if $P(a_{-k}|s,a_k)$ is always positive.
        \end{lemma}
        \begin{proof}
            From the law of total probability, we can infer the equalities:
            \begin{align*}
                P(s'|s,a_k) &= \sum_{a_{-k}\in A_{-k}}P(s'|s,a_k,a_{-k})P(a_{-k}|s,a_k) \\
                            &= \sum_{a_{-k}\in A_{-k}} \mathds{1}_{s'=T(s,a_k,a_{-k})}P(a_{-k}|s,a_k)
            \end{align*}
            Additionally, we know that $\mathds{1}_{s'=T(s,a_k, a_{-k})}$ is always non-negative and we assume that $P(a_{-k}|s,a_k)$ is always positive, Therefore, the relation is true:
            \begin{align*}
                \exists a_{-k}\in A_{-k}, \mathds{1}_{s'=T(s,a_k,a_{-k})}>0 \Rightarrow P(s'|s,a_k) > 0
            \end{align*}
        \end{proof}        
        Based on the given Theorem and Lemmas, we deduce that the strategy $\pi^{ss*}_k$ is deemed optimal when the conditional probability $P(a^{-k}|s,a_k)$ is consistently greater than zero. This assumption essentially means that for any chosen action $a_k$ by agent $k$, there is always a nonzero probability of encountering any other actions $a{-k}$ from other agents. Put simply, regardless of the decision made by agent $k$, there's always a possibility to interact with the varied actions of other agents. To meet this requirement, one can implement an $\epsilon$-greedy strategy with $\epsilon>0$, ensuring that agents occasionally choose actions at random, thereby maintaining a positive probability of diverse actions.
    \subsection{Multiagent Interaction Schemes}
        
        Agents require experience for learning. They collect experience when interacting with other agents. In an independent learning setting, each agent collects its own experience and does not share it with others. Moreover, independent QSS learning agents require an assumption to ensure the individual learned policy $\pi^{ss*}_k$ to be optimal. Therefore, we design our interaction schemes to ensure it.

         \subsubsection{Synchronous Multiagent Interaction (SMA)}\label{sec:IIndependent-Q-Learning}
            From Theorem \ref{pi_ss_k is optimal}, Lemma \ref{the set equivalence} and \ref{equality-prob}, we know that $\pi^{ss*}_k$ is optimal if $P(a_{-k}|s,a_k)$ is always positive. That means $\pi^{ss*}_k$ is optimal if the probability of all possible joint actions should stay positive no matter what action a certain agent executes. Therefore, to ensure $P(a_{-k}|s,a_k)>0$, we need to ensure agents always take a random action with probability $>0$ when interacting with others. More specifically, we can enable agents to execute $\epsilon$-greedy policy, which takes random action with probability $\epsilon>0$, and learned action with probability $1-\epsilon$:
            \begin{align*}
                & \pi^{\epsilon}_k(s) = &\begin{cases}
                                            \pi_k(s) \ &, \textrm{with probability } 1-\epsilon \\
                                            \textrm{a random action} &, \textrm{with probability } \epsilon
                                        \end{cases}
            \end{align*} 

            In synchronous multiagent interaction (SMA), each agent executes $\epsilon$-greedy policy with $\epsilon>0$ and collects experiences to its replay buffers. Then, from its collected experience, agent $k \in K$ would observe $P(a_{-k}|s,a_k)>0$. Therefore, with Theorem \ref{pi_ss_k is optimal}, we know each agent $k \in K$ would learn an individual policy, $\pi^{ss*}_k$, which is optimal.
            
            \begin{algorithm}\label{SMA Interaction}
            	\caption{Sychrnounous Multiagent Interaction} 
            	\begin{algorithmic}[1]
            		\For {$t=1,2,\ldots,t_{max}$}
                        \State All agents execute their $\epsilon$-greedy policy.
                        \State All agents record the interaction to the replay buffers.
            		\EndFor
            	\end{algorithmic} 
            \end{algorithm}

        \subsubsection{Round-Robin Multiagent Interaction (ROMA)}
            When all agents follow the SMA interaction process, each agent can learn an individual optimal policy through independent QSS learning. However, they might not form an optimal joint policy if there exists more than one optimal joint policy in the environment. More specifically, if agents learn their individual optimal policy aiming at reaching a divergent optimal joint policy, the joint of their policies usually ends up being policies, that are not the optimal policy, expected by them. Take a robot coordination task as an example; all robots acknowledge that the optimal solution is that they all turn right or left. Therefore, each of them will learn that the optimal individual action could be either turning left or right. A potential joint of its optimal individual policies would be that some turn left, and some turn right. That is apparently not the result generated by an optimal joint policy. 

            \begin{figure}
                \centering
                \includegraphics[width=0.45\columnwidth]{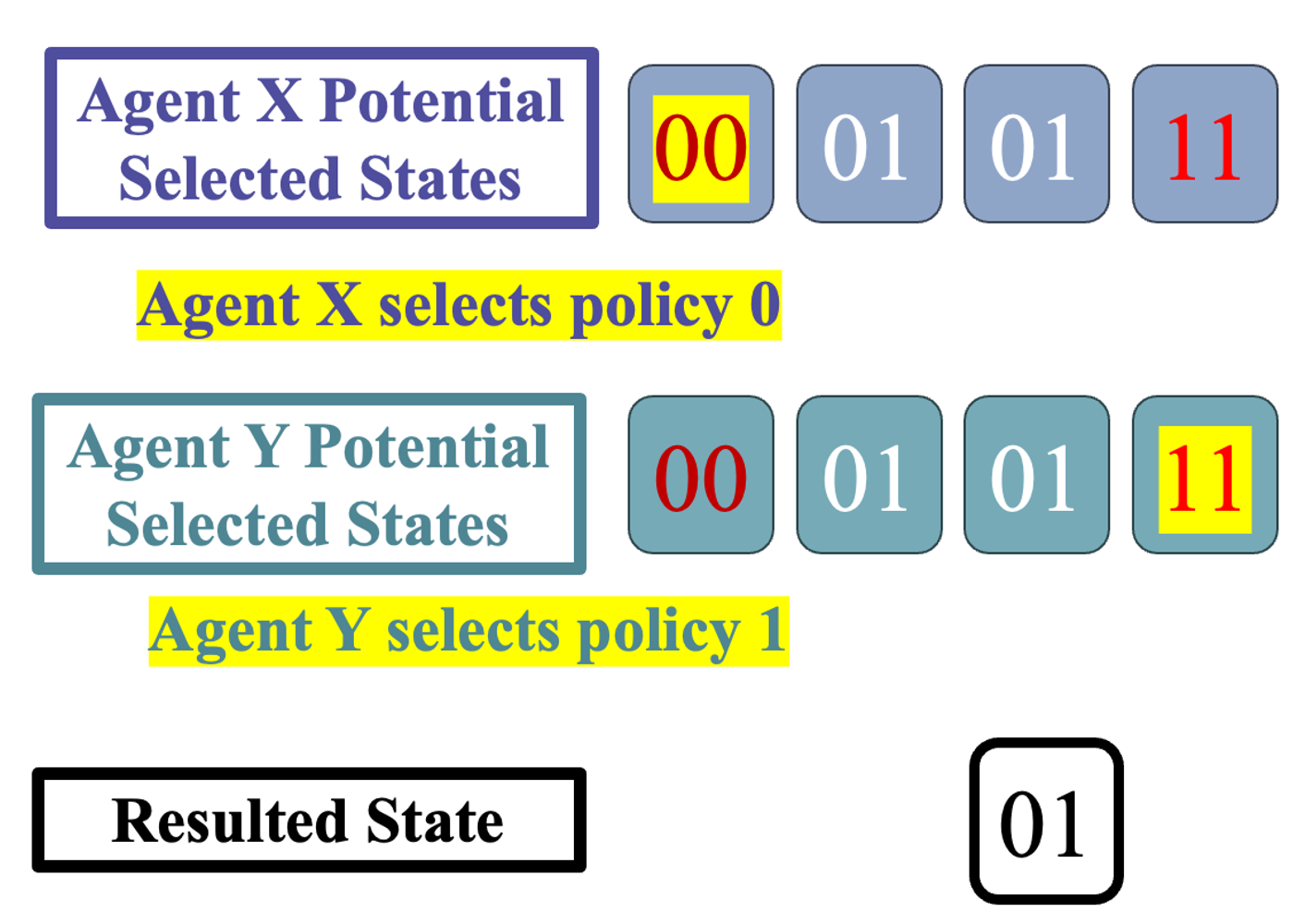}
                \includegraphics[width=0.48\columnwidth]{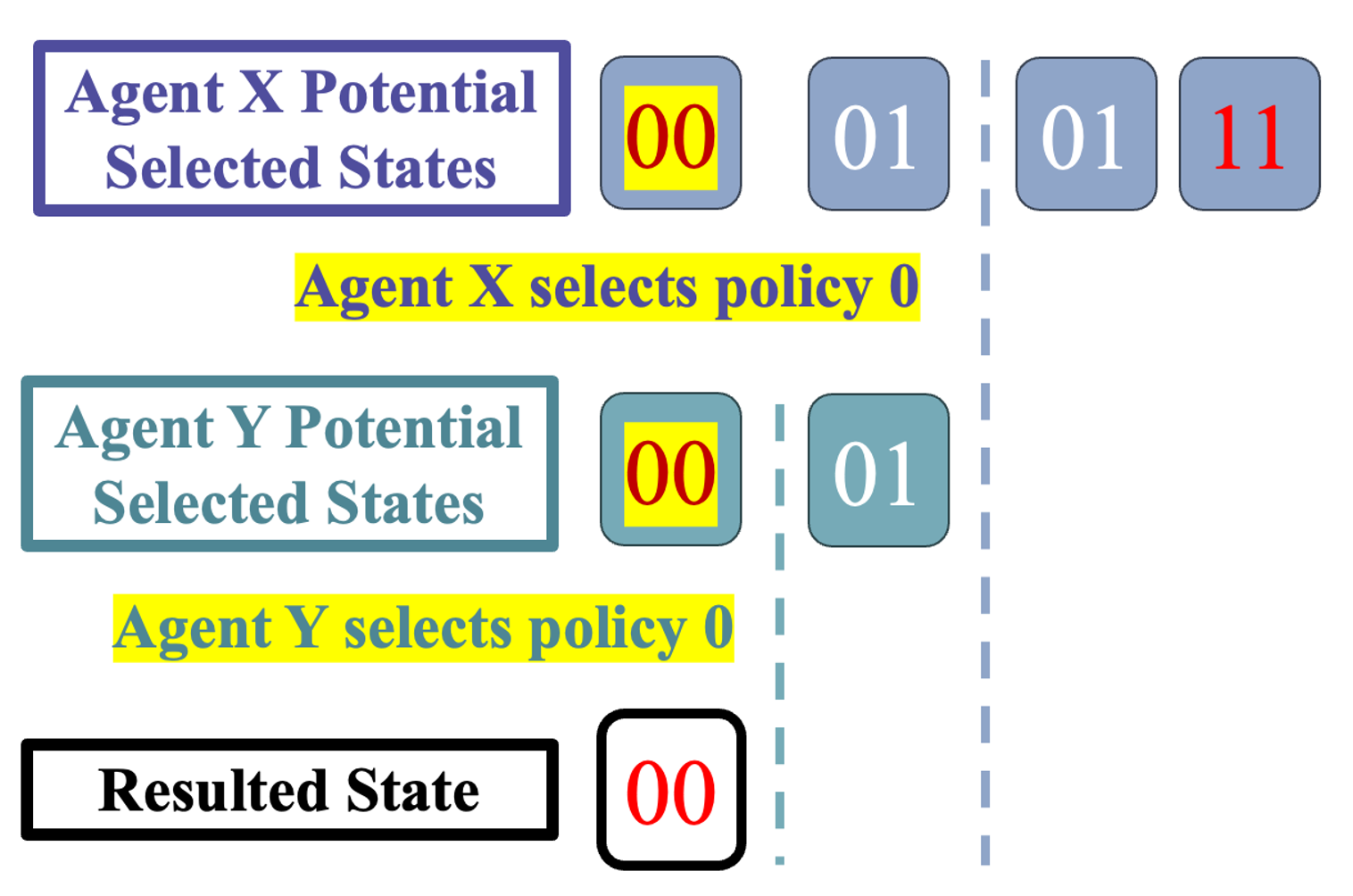}
                \caption{
                    Problem Definition: This is a 2-agent game featuring four potential destination states, each represented by two numbers denoting Agent X's and Agent Y's policies. \textbf{In this game, both agents can identify optimal states, indicated in red text, but they make their decisions independently without knowledge of each other's choices}.
                    \textbf{\textsl{SMA-Scenario(Left):}} Both agents observe all potential states, including two optimal states. Consequently, they might establish divergent objectives, which prevents them from reaching either of the two optimal states. \textbf{\textsl{ROMA-Scenario(Right):}} Agent X observes all potential states and selects its policy accordingly. Agent Y, on the other hand, observes only the states that Agent X's selection can lead to. Consequently, Agent X's selection influences Agent Y's choice, enabling it to align with Agent X's objective and ultimately reach the optimal state.}
                \label{figurelabel}
            \end{figure}
            We design an interaction scheme, round-robin multiagent interaction (ROMA), for agents to collect limited, but sufficient information that facilitates them in aligning their objectives with others so as to cooperate in learning the joint optimal. In ROMA, in each iteration, only one agent can collect experience. We refer to that single agent as the collector. We let a certain agent, $c$, be the only collector for continuous $t_u$ iterations. After agent $c$'s collection period, we let its next agent, which has an index $(c+1)\mod |K|$, be the only collector at the next continuous $t_u$ iterations. We enable agents to repeat the rotation process. 

            During the iterations, where agent $c$ is the only collector, agents with an index, smaller than $c$ are seniors, and agents with an index, greater than $c$, are juniors. Moreover, seniors can only execute their learned policy to interact with others. All other agents, including all juniors and the collector, execute their $\epsilon$-greedy policies, which allows them to take random actions. Consequently, the collector, agent $c$ observes experience, caused and limited by the joint of the learned policies of seniors. In other words, the collector cannot observe the transition conditioned on the joint actions, which its seniors do not take at all. 

            \begin{algorithm}\label{ROMA}
            	\caption{Round-Robin Multiagent Interaction } 
            	\begin{algorithmic}[1]
            		\For {$t=1,2,\ldots,t_{max}$}
                        \State $c = \lceil \frac{t}{t_u} \rceil \mod |K|$  
                        \State $collector$ = agent $c$
                        \State $seniors\ Z_c$ = agents with an index $<c$
                        \State $juniors\ J_c$ = agents with an index $>c$
                        \State 
                            $\forall z\in Z_c$ execute their learned policies, $\forall j\in J_c$ and agent $c$ execute their $\epsilon$-greedy policies
                        \State $collector$ records the interaction to replay buffer
            		\EndFor
            	\end{algorithmic} 
            \end{algorithm}
            
            The observation for a collector is limited to the transitions caused by the joint policy of its seniors. Therefore, if seniors have coordinated in learning a joint optimal, the limited observation is still sufficient to enable the collector to observe the optimal transitions so as to learn an optimal policy. Moreover, the limited observation, caused by seniors, also facilitates the collector to align its objectives with seniors. Consequently, following ROMA, the collector can coordinate with its sensors to learn a joint optimal policy even if multiple optimal joint policies exist. 
            
            As a result, ROMA enables agents to reach optimal coordination. We place the proof of the theorem in the appendix, which can be found on ArXiv.
            \begin{theorem} \label{ROMA-Optimal}
                If all agents follow the ROMA interaction process to collect experiences, the joint of independent QSS learning agents' individual converged strategies, $\pi^{ss*}_k$, is optimal.
            \end{theorem}

    \section{A Practical iQSS Learning Process}
        During the interaction, each agent collects its experience into its replay buffer. From the replay buffer, each agent $k\in K$ can observe multiple experienced tuples, each of which is represented as $(s,a_k,s',r)$ that are the subsequent state $s'$ and the received reward $r$ after taking the action $a_k$ at the state $s$. With the data, each agent can learn its policy by modeling the transition $P(s'|s,a_k)$ and the neighboring states $N(s,a_k)$, identifying the optimal states, optimizing state-based value network $Q(s,s')$, and finally inducing the state-action value function $Q^{ssa}(s,a_k)$, and the policy $\pi^{ss}(s)$.

        Each agent $k\in K$ first learns the observed transition, $P(s'|s,a_k)$, by modeling it using a neural network, represented by parameters $\theta$ and maximizing the likelihood of the observed data point. 
        \begin{align*}
            P(s'|s,a_k) \approx P(s'|s,a_k, \theta)
        \end{align*}
         An agent can use the estimated likelihood to infer the neighboring states. If the likelihood $P(s'|s,a_k,\theta)$ is greater than $\delta$, that means that an agent $k\in K$ observes a subsequent state $s'$ with a probability higher than $\delta$ when it takes action $a_k$ at the state $s$. Therefore, it can consider the subsequent state $s'$ as a possible neighboring state, $s' \in N^{\delta}(s,a_k,\theta)$. Additionally, when $\delta$ is zero, an agent considers all possible states it has observed from the collected data. It is fine if other agents do not change their policy at all. However, if others change their policies, an agent might consider states, which were possible but not going to be reached. Therefore, each agent should set $\delta$ to be close to, but greater than zero. As a result, it would only consider states that are still possible to reach with a probability higher than $\delta$, greater than zero, and thus, would not consider those that will not be reached. 
        \begin{align*}
            N^{\delta}(s,a_k,\theta) \approx \{s'|P(s'|s,a_k,\theta) > \delta\}
        \end{align*}

        Each agent also needs to estimate the state-based value $Q(s,s')$ through the temporal difference learning using \cref{TD-ind-QSS}. When applying it, agents must maximize the state-based value over all possible neighboring states of the state. However, the space of neighboring states is usually large. Therefore, an efficient way to identify optimal neighboring states is desirable. 
        \begin{equation}\label{max_qks}
            \begin{aligned}
                &\max_{s'\in N(s)}Q_k(s, s' ) \\
                &= \max_{a_k\in A_k}\max_{s' \in N(s,a_k)}Q_k(s, s' ) \\
                &= \max_{a_k\in A_k}Q_k(s', \hat{z}_{s,a_k} ) 
            \end{aligned}
        \end{equation}
    
        We can rewrite the maximization using individual neighboring states $N(s,a_k)$ as the second line of the \cref{max_qks}, which maximizes over an individual action space and also the individual neighboring state. The maximization over individual neighboring states $N(s,a_k)$ is computing expensive. Therefore, we develop a method to compute the optimal individual neighboring state $\hat{z}_{s,a_k}$. Consequently, we can compute the maximization by only maximizing over an individual action space.
        \begin{align*}
            \hat{z}_{s,a_k} & = \argmax_{s' \in N_k(s,a_k,\theta)}Q_k(s,s')
        \end{align*}
        The optimal individual neighboring state $\hat{z}_{s,a_k}$ is the state, which provides the max state-based value among all individual neighboring states $N(s, a_k)$. We enable agents to learn $\hat{z}_{s,a_k}$ by an iterative update method. We first assign an agent a random state as $\hat{z}_{s,a_k}$. However, if $P(\hat{z}_{s,a_k}|s,a_k,\theta)$ is smaller than $\delta$, that means $\hat{z}_{s,a_k}$ does not belong to $N(s, a_k)$. Therefore, we must update it by a neighboring state $s'$, which belongs to $N(s, a_k)$, where $P(\hat{s'}_{s,a_k}|s,a_k,\theta)>\delta$. Additionally, if it belongs to $N(s, a_k)$, but offers a worse performance than the other individual neighboring state. Using the better state, we must update $\hat{z}_{s,a_k}$. These two basic principles enable agents to improve their understanding of $\hat{z}_{s,a_k}$ and ultimately learn the optimal neighboring state:
        \begin{align*}
            \hat{z}_{s,a_k} \leftarrow  \begin{cases}
                                            s',    &if\ P(\hat{z}_{s,a_k}|s,a_k,\theta) \leq \delta \\
                                                                & \ and\  P(s'|s,a_k,\theta) > \delta \\
                                            s',    &if\ Q_k(s,\hat{z}_{s,a_k}) < Q_k(s,s') \\
                                                                & \ and\  P(s'|s,a_k,\theta) > \delta \\
                                            \hat{z}_{s,a_k},                 & else
                                        \end{cases}
        \end{align*}
    
        With the optimal neighboring state, we can develop an efficient temporal difference learning method to update the state-based value $Q_k(s, s)$ by replacing the maximization with \cref{max_qks}:
        \begin{equation}\label{Qks-z}
            \begin{aligned}
                Q_k(s, s') \leftarrow
                        &(1-\alpha)Q_k(s,s') + \\
                        &\alpha(r + \gamma \max_{a_k\in A_k}Q_k(s', \hat{z}_{s',a_k} ) )
            \end{aligned}
        \end{equation}

        The temporal difference learning in \cref{Qks-z} enables each agent to update the value $Q^k(s,s')$. In addition, during each iteration, they can use the states-based value $Q(s,s')$ to update the state-action value $Q_k^{ssa}(s, a_k)$ and also further compute the policy $\pi^{ss}_k(s)$:
        \begin{align*}
            Q_k^{ssa}(s, a_k) \leftarrow \max_{s' \in N(s,a_k)}Q_k(s,s')
        \end{align*}
        \begin{align*}
            \pi^{ss}_k(s) \leftarrow \argmax_{a_k \in A_k}Q^{ssa}_k(s, a_k)        
        \end{align*}
        As a result, given enough time each agent learns optimal values and the optimal policy. If an overview of practical iQSS learning is desired, we provide the pseudocode in the appendix, which can be found on ArXiv.
        
    \section{Numerical Experiments}
    We evaluate our algorithms via a coordination challenge that involves independent agents. Every agent begins with knowledge of various strategies but does not know which strategy offers the most group benefit and what strategies other agents are considering. The goal is for agents to tackle two challenges: to independently identify the best strategies and to align their efforts with other agents.

    The issue is divided into three phases, where agents must address two challenges in each to increase collective earnings. Not pinpointing the best strategies leads to low rewards, ranging from $0$ to $10$. Conversely, agents targeting divergent optimal states face hefty fines of $-200$. Achieving success hinges on all agents aligning their goals and collaborating towards common objectives, which can yield significant gains of $100$.

    Agents independently choose from three actions at each stage, forming a joint strategy. Notably, only two specific joint strategies yield successful coordination, directing agents toward one of the two optimal states. Furthermore, successful coordination hinges on all agents converging to pursue the same optimal objective; otherwise, effective coordination is unlikely.

    We simulate coordination among three, five, and seven agents. For three agents, the task is to choose the best state out of 65 states. For five and seven agents, the challenge increases with 100 states. Our investigation compares our novel state-based Q learning technique (iQSS) against the established state-based Q learning method (I2Q) and the conventional action-based Q learning (independent Q Learning). iQSS and I2Q both rely on state values for making decisions but differ in their learning approaches and strategies for selecting the optimal state.

    We assess two interaction strategies: SMA (standard) and ROMA (innovative), applying them across all examined learning approaches. This results in six distinct evaluative comparisons: ROMA-iQSS, SMA-iQSS, ROMA-I2Q, SMA-I2Q, and the independent Q learning techniques including ROMA-indQ and SMA-indQ. We carry out fifteen simulations for each method across scenarios involving $3$, $5$, and $7$ agents to analyze the average group reward and performance indicators. Furthermore, we choose an $\epsilon$-greedy policy setting of $\epsilon=0.8$ to encourage comprehensive exploration. This high $\epsilon$ value is crucial for simulating the objective alignment challenge, ensuring that the majority of agents experience both optimal states. Additionally, adopting a higher $\epsilon$ value helps in minimizing the likelihood of agents becoming entrapped in local optima. Moreover, we show more experiment setup details in the appendix, which can be found on ArXiv.

    The analysis reveals that ROMA significantly diminishes performance variance compared to SMA algorithms. Specifically, in scenarios with three and five agents, the performance of SMA-iQSS demonstrates considerable fluctuations, whereas ROMA-iQSS stabilizes quickly, achieving consistent and superior performance after a few iterations. Furthermore, both ROMA-indQ and ROMA-I2Q exhibit marginally better and notably steadier performance than their SMA equivalents. This reduction in variance can be attributed to the more stationary environment created by the ROMA structure. In SMA settings, all agents are concurrently exploring and making decisions to achieve their goals, leading to challenges in objective alignment when goals differ, thereby producing a non-stationary environment. Conversely, ROMA facilitates a more stable environment by enabling agents to naturally follow the lead of more experienced or "senior" agents in decision-making processes, thereby streamlining objective alignment without the need for explicit discussions.
    
    The plots demonstrate that iQSS, particularly ROMA-iQSS, outperforms I2Q and independent Q learning in efficiently finding optimal states. ROMA-iQSS consistently achieves superior results, indicating its effectiveness in identifying the best state. Although SMA-iQSS does not always deliver top average performance, it reaches significantly higher peaks in some iterations, showing it can find the optimal state despite challenges in aligning objectives. In contrast, I2Q and independent Q learning typically do not attain the significant rewards, emphasizing the advantage of iQSS in state optimization. This advantage stems from I2Q's differing method of state estimation and the limitations of action-based techniques like independent Q learning in achieving similar outcomes.

    \begin{figure}
        \centering        
        \includegraphics[width=0.635\columnwidth]{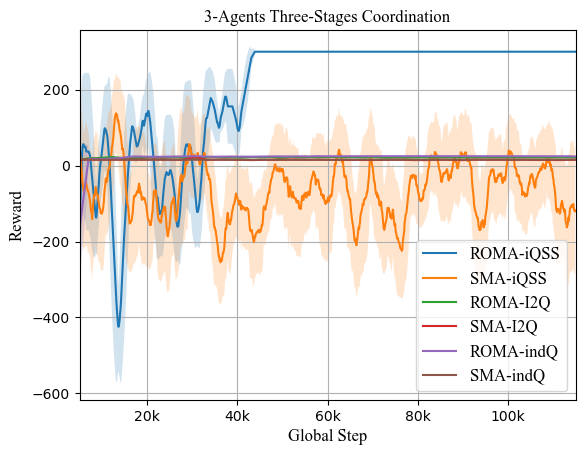}
        \includegraphics[width=0.67\columnwidth]{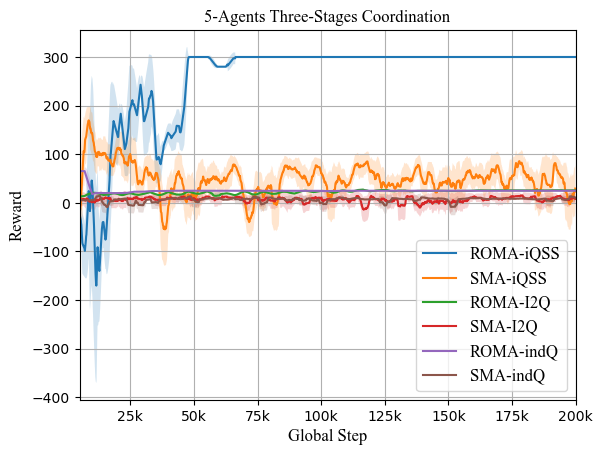}
        \includegraphics[width=0.67\columnwidth]{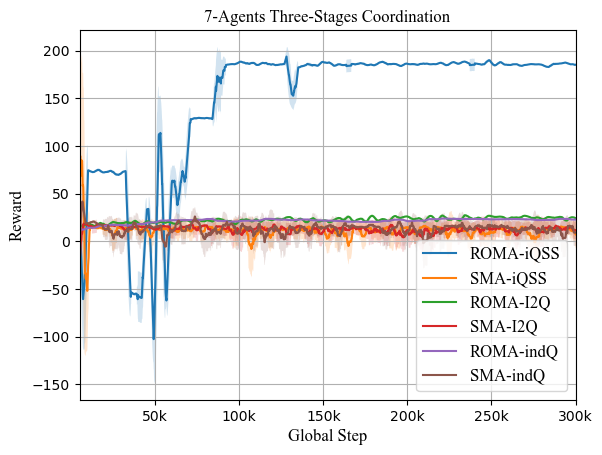}
        \caption{Teams of 3, 5, and 7 agents navigate three-stage coordination, targeting top outcomes in an environment with multiple optimal strategies.}
    \end{figure}
    
    \section{Conclusion and Discussion}
        In this work, we present a novel decentralized framework decentralized framework aimed at overcoming the challenge of forming an optimal collective strategy without direct communication. To address this challenge, we develop two critical components. Firstly, we develop an independent states-based learning method, iQSS, designed to facilitate each agent to efficiently identify optimal states and individual policies. Secondly, we introduce an interaction scheme, ROMA, which empowers agents to collect limited yet sufficient data. This data allows each agent to implicitly guide less experienced peers in their learning journey while simultaneously benefiting from the guidance of more senior agents during their own learning process. As a result, ROMA plays a pivotal role in aligning the objectives of all agents, ultimately leading to successful coordination.
    
        We substantiate our claims with rigorous proof, affirming that iQSS is capable of learning individual optimal policies, while also demonstrating that ROMA fosters objective alignment among agents. Furthermore, we provide empirical evidence through multi-agent coordination simulations, showcasing our method's superiority over the current state-of-the-art state-based value learning approach, I2Q, and the classical action-based, independent Q learning method.
    
        For future work, we aim to extend our multi-agent coordination methods to real-world applications, such as coordinating robots in warehouses or ensuring safe autonomous driving for automobiles. Additionally, we plan to incorporate human-like features into our agents. While human strategies are often inaccessible to robots, our framework equips agents to coordinate effectively with others even without insights into their strategies. Consequently, we are optimistic that our approach holds significant potential to enhance human-robot coordination.



    \bibliographystyle{IEEEtran}
    \bibliography{Ref}

\begin{thebibliography}{10}
\providecommand{\url}[1]{#1}
\csname url@samestyle\endcsname
\providecommand{\newblock}{\relax}
\providecommand{\bibinfo}[2]{#2}
\providecommand{\BIBentrySTDinterwordspacing}{\spaceskip=0pt\relax}
\providecommand{\BIBentryALTinterwordstretchfactor}{4}
\providecommand{\BIBentryALTinterwordspacing}{\spaceskip=\fontdimen2\font plus
\BIBentryALTinterwordstretchfactor\fontdimen3\font minus \fontdimen4\font\relax}
\providecommand{\BIBforeignlanguage}[2]{{%
\expandafter\ifx\csname l@#1\endcsname\relax
\typeout{** WARNING: IEEEtran.bst: No hyphenation pattern has been}%
\typeout{** loaded for the language `#1'. Using the pattern for}%
\typeout{** the default language instead.}%
\else
\language=\csname l@#1\endcsname
\fi
#2}}
\providecommand{\BIBdecl}{\relax}
\BIBdecl

\bibitem{papoudakis2021benchmarking}
\BIBentryALTinterwordspacing
G.~Papoudakis, F.~Christianos, L.~Schäfer, and S.~V. Albrecht, ``Benchmarking {M}ulti-{A}gent {D}eep {R}einforcement {L}earning {A}lgorithms in {C}ooperative {T}asks,'' in \emph{Proceedings of the Neural Information Processing Systems Track on Datasets and Benchmarks (NeurIPS)}, 2021. [Online]. Available: \url{http://arxiv.org/abs/2006.07869}
\BIBentrySTDinterwordspacing

\bibitem{mrc}
J.~J. Koh, G.~Ding, C.~Heckman, L.~Chen, and A.~Roncone, ``{C}ooperative {C}ontrol of {M}obile {R}obots with {Stackelberg} {L}earning,'' in \emph{Proceedings of the 2020 IEEE/RSJ International Conference on Intelligent Robots and Systems (IROS)}, 2020, pp. 7985--7992.

\bibitem{tsc}
X.~Wang, L.~Ke, Z.~Qiao, and X.~Chai, ``Large-scale {T}raffic {S}ignal {C}ontrol {U}sing a {N}ovel {M}ultiagent {R}einforcement {L}earning,'' \emph{IEEE Transactions on Cybernetics}, vol.~51, no.~1, pp. 174--187, 2021.

\bibitem{Wei_2019}
H.~Wei, N.~Xu, H.~Zhang, G.~Zheng, X.~Zang, C.~Chen, W.~Zhang, Y.~Zhu, K.~Xu, and Z.~Li, ``{CoLight}: {Learning Network-Level Cooperation for Traffic Signal Control},'' in \emph{Proceedings of the 28th {ACM} International Conference on Information and Knowledge Management}.\hskip 1em plus 0.5em minus 0.4em\relax {ACM}, Nov 2019.

\bibitem{brawer2023interactive}
J.~Brawer, D.~Ghose, K.~Candon, M.~Qin, A.~Roncone, M.~V{\'a}zquez, and B.~Scassellati, ``{Interactive Policy Shaping for Human-Robot Collaboration with Transparent Matrix Overlays},'' in \emph{Proceedings of the 2023 ACM/IEEE International Conference on Human-Robot Interaction}, 2023, pp. 525--533.

\bibitem{tung2024workspaceopt}
\BIBentryALTinterwordspacing
Y.-S. Tung, M.~B. Luebbers, A.~Roncone, and B.~Hayes, ``{Workspace Optimization Techniques to Improve Prediction of Human Motion During Human-Robot Collaboration},'' in \emph{In Proceedings of the 2024 ACM/IEEE International Conference on Human-Robot Interaction (HRI'24)}, ACM/IEEE.\hskip 1em plus 0.5em minus 0.4em\relax New York, NY, USA: ACM, Mar 2024. [Online]. Available: \url{https://doi.org/10.1145/3610977.3635003}
\BIBentrySTDinterwordspacing

\bibitem{son2019qtran}
K.~Son, D.~Kim, W.~J. Kang, D.~E. Hostallero, and Y.~Yi, ``{QTRAN: Learning to factorize with transformation for cooperative multi-agent reinforcement learning},'' in \emph{International conference on machine learning}.\hskip 1em plus 0.5em minus 0.4em\relax PMLR, 2019, pp. 5887--5896.

\bibitem{gupta2017cooperative}
J.~K. Gupta, M.~Egorov, and M.~Kochenderfer, ``Cooperative {M}ulti-{A}gent {C}ontrol using {D}eep {R}einforcement {L}earning,'' in \emph{Autonomous Agents and Multiagent Systems: AAMAS 2017 Workshops, Best Papers, Revised Selected Papers 16}.\hskip 1em plus 0.5em minus 0.4em\relax Springer, 2017, pp. 66--83.

\bibitem{lowe2020multiagent}
R.~Lowe, Y.~Wu, A.~Tamar, J.~Harb, P.~Abbeel, and I.~Mordatch, ``{Multi-Agent Actor-Critic for Mixed Cooperative-Competitive Environments},'' in \emph{Advances in Neural Information Processing Systems}, 2017.

\bibitem{iqbal2019actorattentioncritic}
S.~Iqbal and F.~Sha, ``{Actor-Attention-Critic for Multi-Agent Reinforcement Learning},'' in \emph{Proceedings of the International Conference on Machine Learning}, 2019.

\bibitem{ming1993multi}
M.~Tan, ``{Multi-agent Reinforcement Learning: Independent vs. cooperative Agents},'' in \emph{Proceedings of the tenth international conference on machine learning}, 1993, pp. 330--337.

\bibitem{dewitt2020independent}
C.~S. de~Witt, T.~Gupta, D.~Makoviichuk, V.~Makoviychuk, P.~H.~S. Torr, M.~Sun, and S.~Whiteson, ``{Is Independent Learning All You Need in the StarCraft Multi-Agent Challenge?}'' 2020.

\bibitem{palmer2018lenient}
G.~Palmer, K.~Tuyls, D.~Bloembergen, and R.~Savani, ``{Lenient Multi-Agent Deep Reinforcement Learning},'' in \emph{Proceedings of the 17th International Conference on Autonomous Agents and Multiagent Systems (AAMAS)}, 2018.

\bibitem{4399095}
L.~Matignon, G.~J. Laurent, and N.~Le~Fort-Piat, ``{Hysteretic {Q}-learning: An algorithm for Decentralized Reinforcement Learning in Cooperative Multi-Agent Teams},'' in \emph{Proceedings of the 2007 IEEE/RSJ International Conference on Intelligent Robots and Systems}, 2007.

\bibitem{I2Q}
J.~Jiang and Z.~Lu, ``{I2Q}: {A Fully Decentralized {Q}-Learning Algorithm},'' \emph{Thirty-Sixth Annual Conference on Neural Information Processing Systems}, 2022.

\bibitem{tan1997multi}
M.~Tan, ``{Multi-Agent Reinforcement Learning: Independent vs. Cooperative Learning},'' \emph{Readings in Agents}, pp. 487--494, 1997.

\bibitem{watkins1992q}
C.~J. Watkins and P.~Dayan, ``Q-learning,'' \emph{Machine Learning}, vol.~8, pp. 279--292, 1992.

\bibitem{QSS}
A.~Edwards, H.~Sahni, R.~Liu, J.~Hung, A.~Jain, R.~Wang, A.~Ecoffet, T.~Miconi, C.~Isbell, and J.~Yosinski, ``Estimating {Q(s,s’)} with {Deep Deterministic Dynamics Gradients},'' in \emph{Proceedings of the 37th International Conference on Machine Learning}, 2020, pp. 2825--2835.

\end{thebibliography}

    \appendix[The proof of Theorem \ref{ROMA-Optimal}]
        
        In Theorem \ref{ROMA-Optimal}, we state:
        If all agents follow the ROMA interaction process to collect experiences, the joint of all agents' individual converged strategies, $\pi^{ss*}_k$, is optimal.
        
        \begin{proof}
            For the readability of the proof, we suppose all agents follow single-round ROMA, where all agents start their rounds to collect experience only if their seniors complete their learning. We can set $t_u = \frac{t}{|K|}$ to satisfy the assumption. Additionally, we also suppose that they stop learning after stop collecting experience. 
            
            After proving the property is true under this assumption, we will consider it in the general multi-round ROMA interaction.

            We want to show that, for all $d\in K$, the joint policy of the first $d$ agents is optimal if all agents follow single-round ROMA to collect experience for independent QSS learning. We show it by induction.

            \textbf{Base Case, d=1:} 
            Given enough time, agent $1$ observes all possible states. We can show it by showing the set of agent $1$;s observed neighboring states including all possible neighboring states:
            \begin{align*}
                & \bigcup_{a_1\in A_1} N(s,a_1)   \\
                & = \bigcup_{a_1\in A_1} \{s'| P(s,a_1,a_{-1})> 0, \forall a_{-1}\in A_{-1}\} \\
                & = \{s'| P(s,a_1,a_{-1})>0, \forall a_{1}\in A_{1}, \forall a_{-1}\in A_{-1}\} \\
                & = \{s'| P(s,a)>0, \forall a \in A\}
            \end{align*}
            The first equivalence is because all agents take random actions with a probability higher than $0$ during agent $1$'s collection turn. The second and the third equivalence comes from the rule of union operations on sets.  
            Given enough time, independent QSS learning enables agent $1$ to learn an individual optimal policy based on the observation of all possible neighboring states.

            \textbf{Base Case, d=2: } 
            Agent $2$ observes the neighboring states:
            \begin{align*}
                \bigcup_{a_2\in A_2} N(s, &a_2)   \\
                = \bigcup_{a_2\in A_2} \{s'|  & P(s, \pi^{ss*}_{<2}(s), a_2, a_{>2})>0, \forall a_{>2}\in A_{>2}\} \\
                = \{s'| &P(s, \pi^{ss*}_1(s), a_2, a_{>2})>0, \forall a_{>2}\in A_{>2})>0,\\
                        & \forall a_{2}\in A_{2}, \forall a_{>2}\in A_{>2}\} \\
                = \{s'| &P(s, \pi^{ss*}_1(s), a_{-1})>0, \forall a_{-1} \in A_{-1}\}
            \end{align*}
            
            The first equivalence is because all agents except agent $1$ take random actions with a probability higher than $0$ and agent $1$ takes actions, suggested by its optimal policy, during agent $2$'s collection turn. The second and the third equivalence comes from the rule of union operations.

            Consequently, agent $2$ observes only the states, which agent $1$'s converged policy $\pi^{ss*}_1$ could cause. However, since its observation depends on agent $1$'s converged policy, which is optimal, agent $2$ can still observe the optimal state transition. As a result, through independent QSS learning, agent $2$ learns its optimal. Moreover, the joint of agent $1$ and agent $2$'s learned policies is also optimal.

            \textbf{Inductive Step: } 
            For some $c\geq2$, we assume that the seniors of agent $c$ have aligned their objectives and coordinated in learning their joint optimal policies. Agent $c$ then observes the states:
             \begin{align*}
                \bigcup_{a_c\in A_c} N(s, &a_c)   \\
                = \bigcup_{a_c\in A_c} \{s'|  & P(s, \pi^{ss*}_{<c}(s), a_c, a_{>c})>0, \forall a_{>c}\in A_{>c}\} \\
                = \{s'| &P(s, \pi^{ss*}_{<c}(s), a_c, a_{>c})>0, \forall a_{>c}\in A_{>c})>0,\\
                        & \forall a_{<c}\in A_{<c}, \forall a_{\geq c}\in A_{\geq c}\} \\
                = \{s'| &P(s, \pi^{ss*}_{<c}(s), a_{-c})>0, \forall a_{-c} \in A_{-c}\}
            \end{align*}
            From the assumption, during agent $c$'s collection turn, its senior execute their joint optimal policies. In addition, all junior agents execute random actions with a probability greater than $0$.
            Therefore, the first equivalence holds. The second and the third equivalence comes from the rule of the union operation on sets.
            
            Agent $c$ observes only the states, its senior's joint policy, $\pi^{ss*}_{<c}$ could cause. However, since its seniors' joint policy is optimal, agent $c$ can still observe the optimal state transition. As a result, agent $c$ learns its optimal through independent QSS learning. Moreover, the joint of agent $c$ and its seniors' learned policies is also optimal.

            \textbf{Results of the Induction:} 
            Our proof shows the statement is true. It implies that the joint policy of all agents is optimal if all agents follow single-round ROMA to collect experience for independent QSS learning. 

            \textbf{Discussion on the multi-round ROMA:}
            We can adapt the proof from single-round ROMA to multi-round ROMA by replacing the lower bound, $0$, of the individual transition by $\delta$, which is close to, but greater than $0$. In other words, the neighboring state observed by agent $k$ is $N^R(s,a_k)$ instead of $N(s,a_k)$:
            \begin{align*} \label{NR}
                N^R(s,a_k) = \{s'|P(s'|s,a_k) > \delta\}
            \end{align*}
            It is because, in multi-round ROMA, agents observe a transition, which is not stationary if its seniors change their policies during their learning. The real transition probability of a certain subsequent state may be large at the start but may become zero if senior agents change their policy to favor other states. In that scenario, the observed transition probability is non-stationary at the start but should decrease gradually to a value close to zero after seniors complete their learning. 

            To filter the choices, dropped by seniors, the collector should consider states with a transition probability higher than a certain small number, which is close to, but greater than zero. Consequently, the collector would only consider the states, its senior's current policy could lead to. As a result, we can show that the joint policy of it and its seniors is optimal following a proving process, similar to the proof in single-round ROMA.
        \end{proof}
        \appendix[Emprical Experiments Setting on ROMA: Early-Stopping and Pre-Collection Mechanisms]{
            Our empirical studies reveal that two mechanisms, Early-Stopping and Pre-Collection, show enhanced learning stationarity when the ROMA interaction scheme is applied. 
            
            \textbf{Early-Stopping: }
                In the Early-Stopping mechanism, more experienced agents conclude their learning process ahead of their junior counterparts. Specifically,we divide the total iteration count into $|K|$ partitions, with $|K|$ representing the number of agents. Each section consists of $t^*$ iterations. An agent indexed by $c$ discontinues its learning after reaching $c\ t^*$ iterations. 
                
                This mechanism enhances the learning environment's stability for less experienced agents because it relies on a straightforward principle: once senior agents cease their learning, they solidify their strategies.
                
            \textbf{Pre-Collection: }
                In the traditional ROMA framework, agents collect experiences exclusively during their own turns, leading to a scenario where each agent learns from a unique set of experiences. Such a method can render their learning less synergistic with their peers'. For example, if some agents mainly investigates the left side of the environment while the rest concentrate on the right, both groups acquire specific insights. However, this disparate learning fails to contribute to formulating an optimal joint strategy for either environment section. 

                To address these challenges, a simple adjustment can be made: allowing agents to collect experiences not just on their own turns, but also during the turn immediately preceding theirs. This change means that the experiences of their peers also get included in their experience replay buffer, encouraging agents to learn from the same experiences and align their learning with one another. Furthermore, implementation details on this are specified in Line 5, 9 and 13 of the pseudocode in the Algorithm titled \textit{ROMA with Early Stopping and Pre-Collection Mechanism}.

            \begin{algorithm}\label{ROMA ES and PC}
            	\caption{ROMA with Early Stopping and Pre-Collection Mechanism} 
            	\begin{algorithmic}[1]
                    \State $t^{*} = t_{max}/|K|$ 
                    \State $t_u = t^{*}/n_{rounds}$
            		\For {$t=1,2,\ldots,t_{max}$}
                        \State $c = \lceil \frac{t}{t_u} \rceil \mod |K|$      
                        \State $nc = (c+1) \mod |K|$
                        \State  \If{$t<c\ {t^{*}}$}     
                                    \State $collector$ = agent $c$
                                    \State $next-collector$ = agent $nc$
                                    \State $seniors\ Z_c$ = agents with an index $<c$
                                    \State $juniors\ J_c$ = agents with an index $>c$
                                    \State 
                                        $\forall z\in Z_c$ execute their learned policies, $\forall j\in J_c$ and agent $c$ execute their $\epsilon$-greedy policies.
                                    \State $collector$ and $next-collector$ record the interaction to replay buffer.
                                \EndIf
                                
            		\EndFor
            	\end{algorithmic} 
            \end{algorithm}
        }

        \appendix[An overview of A Practical iQSS Learning]
            Each iQSS agent carries out Practical iQSS Learning with its collected dataset $D_k$. Initially, the agent initializes its models, including value functions $Q_k$, $Q_k^{ssa}$, the policy $\pi^{ss}_k(s)$, along with the parameters for the state-likelihood model $\theta_k$, and the parameters for the best-state estimator $\omega_k$.

            During training iterations, agent $k$ selects a batch of experiences $(s,a_k,s',r)$, where $s$ represents the current states, $a_k$ the actions, $s'$ the subsequent states, and $r$ the rewards. With this data, agent $k$ updates $\theta_k$ using maximum likelihood estimation as indicated on Line 5. It also refines its best-state estimator in line with the update strategies discussed in "A Practical iQSS Learning Process", specifically concerning $\hat{z}_{s,a_k}$. Moreover, agent $k$ revises the state-based value function $Q_k$ through temporal difference learning. Following this, agent $k$ uses the derived state-based value to enhance the action-based value $Q_k^{ssa}$ and to further deduce the policy $\pi^{ss}_k(s)$.
                    
            \begin{algorithm}\label{iQSS}
                \caption{Agent $k$: iQSS with Optimal State Identification} 
                \begin{algorithmic}[1]
                    \State Initialize $Q_k$, $Q_k^{ssa}$, $\pi^{ss}_k(s)$, $\theta_k$, and $\omega_k$.
                    \State Initialize the experience buffer $D_k$
                    \For {$t=1,2,\ldots,t_{max}$}
                        \State $(s,a_k,s',r) \sim D_k$
                        \State $\theta_k \leftarrow \argmax_{\theta}  \log P(s'|s,a_k,{\theta})$ 
                        \State $z = \hat{z}_{s,a_k}$
                        \If{ $P(z|s,a_k,{\theta})\leq\delta$ or $Q(s,z)\leq Q(s,s')$}
                            \If{ $P(s'|s,a_k,{\theta})>\delta$}
                                \State $\hat{z}_{s,a_k} \leftarrow s'$
                            \EndIf
                        \EndIf
                                \begin{align*}
                                    Q_k(s, s) \leftarrow
                                            &(1-\alpha)Q_k(s,s) + \\
                                            &\alpha(r + \gamma \max_{a_k\in A_k}Q_k(s', \hat{z}_{s,a_k} ) )
                                \end{align*}
                        \State $Q_k^{ssa}(s, a_k) \leftarrow \max_{s' \in N(s,a_k)}Q_k(s,s')$
                        \State $\pi^{ss}_k(s) \leftarrow \argmax_{a_k \in A_k}Q^{ssa}_k(s, a_k)$        
                    \EndFor
                \end{algorithmic} 
            \end{algorithm}     
    
    



    \addtolength{\textheight}{-12cm}   
    
\end{document}